\documentclass[12pt]{amsart}
\usepackage{amsmath, amsthm, amsfonts, amsbsy, amssymb, upref, enumerate, bigstrut,  color, mathtools, mathrsfs, float, bm}

\usepackage[left=1.3in,top=1.3in,right=1.3in]{geometry}

\usepackage{epsfig, graphicx}
\usepackage{ hyperref}

\newtheorem{theorem}{Theorem}[section]

\newtheorem{proposition}[theorem]{Proposition}

\newcommand{\ra}{\rightarrow}

\newcommand{\smallavg}[1]{\langle #1 \rangle}

\newcommand{\nto}{\mbox{$\;\longrightarrow_{\hspace*{-0.42cm}{\small n}}\;$~}}

\newtheorem{lemma1}{Lemma} 
\newtheorem{thm1}{Theorem} 

\begin{document}
\title[No Replica Symmetry Breaking]
{{Absence of Replica Symmetry Breaking in the  
		Random Field Mixed-Spin Ginzburg-Landau Model
}}
\author{R. Vila}
\address{
\newline 
Departamento de Estat\'istica -
Universidade de Bras\'ilia, Brazil,
Email: \textup{\tt rovig161@gmail.com}
}

\date{\today}

\keywords{Random Field Ginzburg-Landau Model $\cdot$ Replica symmetry.}
\subjclass[2010]{MSC 82B20, MSC 82B44, MSC 60K35.}

\begin{abstract}
In this paper, an extension of the random field Ginzburg-Landau model 
on the hypercubic lattice is considered by adding $p$-spin ($p\geqslant 2$) interactions coupled to general disorders. This new model is called the random field mixed-spin Ginzburg-Landau model.
We proved that, in the infinite volume limit of this model, the variance of spin overlap vanishes. 
\end{abstract}

\maketitle
\section{Introduction}
A long-standing debate on the presence of elusive spin glass phase (i.e., when the ferromagnetic susceptibility is finite, while the spin glass susceptibility is infinite) in several systems has been the subject of research for many years. For example, in \cite{KRZ-10,KRSZ-11}, among other things, the authors argued that there is no spin glass phase at equilibrium in the random field Ising model, the random field Ginzburg-Landau model and the random temperature Ginzburg-Landau model, for non-negative interactions and arbitrary disorders on an arbitrary lattice.
In 2015, Chatterjee \cite{chatterjee2015disorder} provided a rigorous
mathematical proof supporting part (absence of replica symmetry breaking) of the findings claimed in \cite{KRZ-10,KRSZ-11} for the random field Ising model, in the case of Gaussian disorders. In 2019, Itoi and Utsunomiya \cite{Itoi-2019} extended the Chatterjee results for a disordered spin model (the random field Ginzburg-Landau model) in which the disorder is Gaussian and the single-spins are not limited. 
More recently, Roldan and Vila (2020) \cite{RV2020} have extended the Chatterjee results for the random field Ising model in weak disorders.

The objective of this paper is to prove the absence of replica symmetry breaking in the random field Ginzburg-Landau model allowing   $p$-spin interactions which are coupled to  disorders not necessarily Gaussian. This new model is called the random field mixed-spin Ginzburg-Landau model. Results involving absence of replica symmetry breaking in the random field Ginzburg-Landau model and related ones can be found in \cite{chatterjee2015disorder,Itoi2019,Itoi-2019,RV2020}. Roughly speaking, by adding asymptotically vanishing perturbations to the Hamiltonian corresponding to Gibbs measure of the random field Ginzburg-Landau model, we prove that, in the thermodynamic limit of this model, the degree to which the values of a spin overlap differ from its expectation value vanishes. So, in our paper, as the ones in \cite{chatterjee2015disorder,Itoi-2019}, we give a support in favor of the absence of the spin glass phase, but for general disorders.  

This paper is organized as follows. In Section 2, we present the random field mixed-spin Ginzburg-Landau model, and some definitions related to this model. 
In Section 3 we state the main result of paper. In Section 4, we present the proof of the main result in details.

\section{The model and some definitions}

The model is defined as follows.
Given $n\geqslant 1$, let 
$V_n=\mathbb{Z}^d\cap[1,n]^d$, $d\geqslant 1$, be a finite subset of vertices of 
$d$-dimensional hypercubic lattice 
with cardinality denoted by $|V_n|$.
The 
Hamiltonian of the random field mixed-spin Ginzburg-Landau model on the set of spins configurations  $\mathbb{R}^{V_n}$ reads
\begin{align}\label{ham-1}
H_{n;\boldsymbol{\alpha}}(\phi)
=
H_n(\phi;g)
+
u\sum_{x\in V_n} \phi_x^4
-
r \sum_{x\in V_n} \phi_x^2
-
H^{\rm per}_{n;\boldsymbol{\alpha}}(\phi;\xi),
\end{align}
where the Hamiltonian $H_n(\phi;g)$ 
is
given by
\begin{align}\label{Hamiltonian-g}
H_n(\phi;g)
=
-	
\beta \sum_{\langle xy \rangle}\phi_x \phi_y 
- 
h\sum_{x\in V_n} g_x \phi_x.
\end{align}
Here, $\langle xy \rangle$ denotes the set of ordered pairs in $V_n$ of nearest neighbors, $\beta>0$ and $h>0$ are called inverse temperature and field strength, respectively. Furthermore,
$u>0$, $r\in\mathbb{R}$ and $\phi=(\phi_x)\in \mathbb{R}^{V_n}$ is a spin configuration.
In this paper, the random external magnetic field
$g=(g_x)$ (called disorder) 
consists of independent and identically distributed (i.i.d.) standard Gaussian random variables. Following \cite{RV2020}, it is not trivial to remove the Gaussian hypothesis on $g$ because a generalization of the Gaussian integration by parts for not limited spins is needed which is not always possible to obtain.

The perturbing Hamiltonian $H^{\rm per}_{n;\boldsymbol{\alpha}}(\phi;\xi)$ in \eqref{ham-1} is defined in function of a pure $p$-spin Hamiltonian $H_{n;p}(\phi;\xi)$, for $p\geqslant 2$, as follows
\begin{align}\label{Function-H}
\begin{array}{lllll}
&\displaystyle
H^{\rm per}_{n;\boldsymbol{\alpha}}(\phi;\xi)
=
i(1-\delta_{\boldsymbol{\alpha},\boldsymbol{0}})
c_n\sum_{p=2}^\infty \alpha_p 2^{-p}
H_{n;p}(\phi;\xi);
\\[0,5cm]
&\displaystyle
H_{n;p}(\phi;\xi)
=
{1\over\vert V_n\vert^{(p-1)/2}}
\sum_{{\bf x}\in \otimes_{l=2}^p V_n} \xi_{\bf x} \phi_{\bf x};
\end{array}
\end{align}
where $\delta$ is the Kronecker delta function, $i$ is the imaginary unit, and
\begin{align*}
\phi_{\bf x}= \phi_{x_1}\cdots \phi_{x_p},
\end{align*}
for ${\bf x}=(x_1,\ldots,x_p)$.
Here in \eqref{Function-H} and throughout this paper,
the sequences of real numbers $(c_n)$  and
$\boldsymbol{\alpha}=(\alpha_p)$ satisfy $c_n\nto 0$ and $\vert\alpha_p\vert\leqslant 1$, respectively. In order that the Hamiltonian 
$H_{n;\boldsymbol{\alpha}}(\phi)$ in \eqref{ham-1} to be well-defined almost surely (see Proposition \ref{prop-def-ham}), we  assume the following regularity condition:
\begin{align}\label{regularity condition}
\sum_{p=2}^\infty \alpha_p 2^{-p} C_{2p}^{1/2}
<\infty,
\end{align}
where the sequence $(C_{k})$ is formed by  upper bounds of expectation of the function of a spin variable whose existence is guaranteed by Lemma II.1 of \cite{Itoi-2019} (see  Proposition \ref{Itoi-bound}).
The random sequence $\xi=(\xi_{\bf x})$
consists of i.i.d. real-valued random variables $\xi_{\bf x}$ with zero-mean and unit-variance. 
Moreover,
the disorders $g=(g_x)$ and $\xi=(\xi_{\bf x})$ are assumed to be independent of each other.

The Gibbs measure of the random field mixed-spin Ginzburg-Landau model on 
$\mathbb{R}^{V_n}$ is given by
\begin{align}\label{gibbs-measure2}
{G}_{n;\boldsymbol{\alpha}}(\{\phi\})
=
{1\over Z_{n;\boldsymbol{\alpha}}}
\exp\big(
-H_{n;\boldsymbol{\alpha}}(\phi)
\big), \quad \phi=(\phi_x)\in \mathbb{R}^{V_n},
\end{align}
where $H_{n;\boldsymbol{\alpha}}(\phi)$ is as in \eqref{ham-1}.
Note that the classical random field Ginzburg-Landau
model 
is recovered 
when
$\boldsymbol{\alpha}=\boldsymbol{0}$. 
Denoting ${\rm D}\phi= \prod_{x\in V_n}\exp(-u\phi_x^4+r\phi_x^2) \, {\rm d} \phi_x$, 
the constant $Z_{n;\boldsymbol{\alpha}}$ is the partition function and it is given by
\begin{align}\label{partition-function}
Z_{n;\boldsymbol{\alpha}}
=
\int_{\mathbb{R}^{|V_n|}}  
\exp\big(-	
H_n(\phi;g)+H^{\rm per}_{n;\boldsymbol{\alpha}}(\phi;\xi)
\big) \, {\rm D}\phi.
\end{align}
\smallskip 

For a function $f:(\mathbb{R}^{V_n})^m\to\mathbb{R}$, $m\geqslant 1$, we define
\begin{align*}
\langle \,f\,\rangle_{\boldsymbol{\alpha}}
&\coloneqq
\dfrac{1}{Z_{n;\boldsymbol{\alpha}}^m}
\int
f(\phi^1,\ldots,\phi^m) 
\exp\left(\sum_{s=1}^{m}\biggl(
-	
H_n(\phi^s;g)
+
H^{\rm per}_{n;\boldsymbol{\alpha}}(\phi^s;\xi)
\biggr)\right) 
{\text{D}\phi^1}\cdots {\text{D}\phi^m}.
\end{align*}

Following the notation of Talagrand \cite{talagrand2003spin}, we write
$ 
\mathbb{E}\langle \,f\,\rangle_{\boldsymbol{\alpha}}
$
to denote the average of the Gibbs expectation $\langle \,f\,\rangle_{\boldsymbol{\alpha}}$  over all the realizations of the disorder $(g_x, \xi_{\bf x})$, for $p\geqslant 2$.


\section{The main result}
If $\phi^1,\phi^2,\ldots$ are i.i.d
configurations under (random) Gibbs measure \eqref{gibbs-measure2}, known as replicas, the  spin overlap
between two replicas $\phi^\ell$, $\phi^{\ell'}$ is defined as
\begin{align}\label{overlap}
R_{\ell,\ell'}
\coloneqq 
{1\over |V_n|}\sum_{x\in V_n} 
\phi^\ell_x \phi^{\ell'}_x, \quad \forall \ell,\ell'.
\end{align}
Note that $R_{\ell,\ell}\neq 1$.

%


The system is said to exhibit replica symmetry breaking (as stated by Chatterjee
in the random field Ising model \cite{chatterjee2015disorder}) when a finite variance calculated in the replica symmetric expectation, in the infinite volume limit, is observed
\[
\lim_{n\to\infty}
\mathbb{E}\big\langle (R_{1,2}-\mathbb{E}\langle R_{1,2}\rangle_{\boldsymbol{\alpha}} )^2\big\rangle_{\boldsymbol{\alpha}}
> 0.
\]

In the present paper, 
the next main theorem shows that
this does not happen in the random field mixed-spin Ginzburg-Landau model.

\begin{thm1}
	\label{rsbthm}
	For almost all coupling constants $(\beta, h,u,r)\in (0, \infty)^3\times\mathbb{R}$, in the 
	random field mixed-spin Ginzburg-Landau model, the following variance vanishes
	\[
	\mathbb{E}\big\langle (R_{1,2}-\mathbb{E}\langle R_{1,2}\rangle_{\boldsymbol{\alpha}} )^2\big\rangle_{\boldsymbol{\alpha}}
	\underset{ n\rightarrow\infty}{\xrightarrow{\hspace*{0,9cm}}} 0
	\]
	in the infinite volume limit.
\end{thm1}
The rest of this paper is devoted to the proof of Theorem \ref{rsbthm}.

\section{Proof}
Before beginning the proof details of the main theorem of this paper, 
we will enunciate and prove some results that will 
facilitate the presentation of paper.

\subsection{Upper bound for the expectation of the function of a spin variable}
Our first tool will be the following proposition. Its proof appears in Itoi and Utsunomiya (2019) \cite{Itoi-2019}, Lemma II.1.
\begin{proposition}\label{Itoi-bound}
	There exists some positive constant $C_k=C_k(h)$ independent of the system size $n$ so that
	\begin{align*}
	\mathbb{E}\langle\phi_{ x}^k\rangle_{\boldsymbol{0}}\leqslant C_k \ \quad \text{for all} \ {x}\in V_n,
	\end{align*}
	where $k$ is an arbitrary even  integer.
\end{proposition}

\begin{proposition}\label{prop-bound-H}
	The expectation of the	function of a pure $p$-spin Hamiltonian $H_{n;p}(\phi;\xi)$ has an upper bound
	\begin{align*}
	\mathbb{E}\big\langle [H_{n;p}(\phi;\xi)]^2 \big\rangle_{\boldsymbol{0}}\leqslant
	\vert V_n\vert C_{2p}, \quad p\geqslant 2.
	\end{align*}
\end{proposition}
\begin{proof}
	Let $\widehat{\mathbb{E}}$ be the expectation over $\xi_{\bf x}$,  $p\geqslant 2$. By Fubini's theorem we have
	\begin{align*}
	\mathbb{E}\big\langle [H_{n;p}(\phi;\xi)]^2 \big\rangle_{\boldsymbol{0}}
	=
	\mathbb{E}\big\langle \widehat{\mathbb{E}}[H_{n;p}(\phi;\xi)]^2 \big\rangle_{\boldsymbol{0}}
	=
	\mathbb{E}\biggl\langle
	{1\over\vert V_n\vert^{p-1}}
	\sum_{{\bf x}}  \phi_{\bf x}^2
	\biggr\rangle_{\boldsymbol{0}},
	\end{align*}
	where in the last equality we used that $\xi=(\xi_{\bf x})$
	consists of i.i.d. real-valued random variables with $\widehat{\mathbb{E}}\xi_{\bf x}=0$ and $\widehat{\mathbb{E}}\xi_{\bf x}^2=1$. By using Hölder's inequality and Proposition \ref{Itoi-bound}, the above expression is at most
	\begin{align*}
	{1\over\vert V_n\vert^{p-1}}
	\sum_{{\bf x}} 
	\prod_{j=1}^{p}
	\mathbb{E}^{1/p}
	\langle \phi_{x_j}^{2p} \,
	\rangle_{\boldsymbol{0}}
	\leqslant
	\vert V_n\vert C_{2p}.
	\end{align*}
	Then the proof follows.
\end{proof}

Our second tool will be the following lemma.
This result plays a fundamental role in our work, because it allows us to relate the averages of both Gibbs expectations corresponding to (classic) random field Ginzburg-Landau model (when $\boldsymbol{\alpha}=\boldsymbol{0}$) and to our model \eqref{gibbs-measure2}.
\begin{lemma1}\label{prop-fund}
	Let $F(\phi):\mathbb{R}^{V_n} \to \mathbb{R}$ be a non-negative Borel function. There is a positive constant $\widetilde{C}=\widetilde{C}(n)$, with $\widetilde{C}\nto 1$ as $n\to\infty$, so that
	\begin{align*}
	\mathbb{E}\langle  F(\phi)\rangle_{\boldsymbol{\alpha}}
	\leqslant
	\widetilde{C}\, 
	\mathbb{E}^{1/2}
	\langle  F^2(\phi) \,
	\rangle_{\boldsymbol{0}}.
	\end{align*}
\end{lemma1}
\begin{proof}
	A straightforward computation shows that
	\begin{align*}
	\mathbb{E}\langle F(\phi)\rangle_{\boldsymbol{\alpha}}
	&=
	\mathbb{E} \,
	{Z_{n;\boldsymbol{0}}\over Z_{n;\boldsymbol{\alpha}}} \, 
	\big\langle F(\phi) \,
	{\exp}\big({H^{\rm per}_{n;\boldsymbol{\alpha}}(\phi;\xi)}\big)
	\big\rangle_{\boldsymbol{0}}
	\\[0,15cm]
	&=
	\mathbb{E} \,
	{
		\big\langle F(\phi) \,
		{\exp}\big({H^{\rm per}_{n;\boldsymbol{\alpha}}(\phi;\xi)}\big)
		\big\rangle_{\boldsymbol{0}}
		\over 
		\big\langle
		{\exp}\big({H^{\rm per}_{n;\boldsymbol{\alpha}}(\phi;\xi)}\big)
		\big\rangle_{\boldsymbol{0}}
	}.
	\end{align*}
	By using Jensen and  Cauchy-Schwarz inequalities, the above expression is
	\begin{align*}
	&\leqslant
	\mathbb{E} \,
	{
		\big\langle  F(\phi) \,
		{\exp}\big({H^{\rm per}_{n;\boldsymbol{\alpha}}(\phi;\xi)}\big)
		\big\rangle_{\boldsymbol{0}}
		\over 
		{\exp}
		\big\langle
		\big({H^{\rm per}_{n;\boldsymbol{\alpha}}(\phi;\xi)}\big)
		\big\rangle_{\boldsymbol{0}}
	} 
	\\[0,2cm]
	&=
	\mathbb{E} \,
	\Big\langle  F(\phi) \,
	{\exp}\Big(
	{H^{\rm per}_{n;\boldsymbol{\alpha}}(\phi;\xi)}
	-
	\big\langle
	{H^{\rm per}_{n;\boldsymbol{\alpha}}(\phi;\xi)}
	\big\rangle_{\boldsymbol{0}}
	\Big)
	\Big\rangle_{\boldsymbol{0}} 
	\\[0,2cm]
	&\leqslant
	\mathbb{E}^{1/2}
	\langle  F^2(\phi) \,
	\rangle_{\boldsymbol{0}}
	\,
	\mathbb{E}^{1/2} 
	\Big\langle
	{\exp}\Big(
	2\big[
	{H^{\rm per}_{n;\boldsymbol{\alpha}}(\phi;\xi)}
	-
	\big\langle
	{H^{\rm per}_{n;\boldsymbol{\alpha}}(\phi;\xi)}
	\big\rangle_{\boldsymbol{0}}
	\big]
	\Big)
	\Big\rangle_{\boldsymbol{0}}.
	\end{align*}
	Therefore,
	\begin{align} \label{main-relation}
	\mathbb{E}\langle F(\phi)\rangle_{\boldsymbol{\alpha}}
	\leqslant
	\mathbb{E}^{1/2}
	\langle  F^2(\phi) \,
	\rangle_{\boldsymbol{0}}
	\,
	\mathbb{E}^{1/2} 
	\Big\langle
	{\exp}\Big(
	2\big[
	{H^{\rm per}_{n;\boldsymbol{\alpha}}(\phi;\xi)}
	-
	\big\langle
	{H^{\rm per}_{n;\boldsymbol{\alpha}}(\phi;\xi)}
	\big\rangle_{\boldsymbol{0}}
	\big]
	\Big)
	\Big\rangle_{\boldsymbol{0}}.
	\end{align}
	
	On the other hand, let 
	\begin{align}\label{def-X}
	X_n\coloneqq (1-\delta_{\boldsymbol{\alpha},\boldsymbol{0}})\sum_{p=2}^\infty \alpha_p 2^{-p}
	H_{n;p}(\phi;\xi).
	\end{align} 
	Note that, ${H^{\rm per}_{n;\boldsymbol{\alpha}}(\phi;\xi)}=i c_n X_n$.
	By triangle inequality, 
	$
	\mathbb{E} 
	\big\langle
	\big\vert
	X_n
	-
	\langle
	X_n
	\rangle_{\boldsymbol{0}}
	\big\vert
	\big\rangle_{\boldsymbol{0}}
	\leqslant
	2
	\mathbb{E} 
	\big\langle
	\vert
	X_n
	\vert
	\big\rangle_{\boldsymbol{0}},
	$
	with
	\begin{align}\label{finite-exp}
	\mathbb{E}
	\big\langle	
	\vert X_n\vert
	\big\rangle_{\boldsymbol{0}} 
	&\leqslant
	(1-\delta_{\boldsymbol{\alpha},\boldsymbol{0}})
	\sum_{p=2}^\infty \alpha_p 2^{-p}
	\mathbb{E}
	\big\langle	
	\vert H_{n;p}(\phi;\xi) \vert 
	\big\rangle_{\boldsymbol{0}} \nonumber
	\\[0,15cm]
	&\leqslant
	\sum_{p=2}^\infty \alpha_p 2^{-p}
	\mathbb{E}^{1/2}
	\big\langle	
	[H_{n;p}(\phi;\xi)]^2 
	\big\rangle_{\boldsymbol{0}} \nonumber
	\\[0,15cm]
	&\leqslant
	\sqrt{\vert V_n\vert} \sum_{p=2}^\infty \alpha_p 2^{-p} C_{2p}^{1/2}
	<\infty,
	\end{align} 
	where in the second inequality we use Lyapunov's inequality and in the third one the Proposition \ref{prop-bound-H} with  the regularity condition \eqref{regularity condition}.

	Consequently, since $\mathbb{E} 
	\big\langle
	\big\vert
	X_n
	-
	\langle
	X_n
	\rangle_{\boldsymbol{0}}
	\big\vert
	\big\rangle_{\boldsymbol{0}}<\infty$, it follows that,  for $n\to\infty$,
	\begin{align}\label{stat-2}
	\mathbb{E}
	\Big\langle
	{\exp}\Big(
	2\big[
	{H^{\rm per}_{n;\boldsymbol{\alpha}}(\phi;\xi)}
	-
	\big\langle
	{H^{\rm per}_{n;\boldsymbol{\alpha}}(\phi;\xi)}
	\big\rangle_{\boldsymbol{0}}
	\big]
	\Big)
	\Big\rangle_{\boldsymbol{0}} \nonumber
	&=
	\mathbb{E}
	\big\langle 
	{\exp}\big(2ic_n [X_n - \langle X_n\rangle_{\boldsymbol{0}}]\big)
	\big\rangle_{\boldsymbol{0}} 
	\\[0,2cm]
	&=
	1+2ic_n \mathbb{E}
	\langle X_n- \langle X_n\rangle_{\boldsymbol{0}} \rangle_{\boldsymbol{0}} 
	+
	2ic_n o(1)
	\nonumber
	\\[0,2cm]
	&=
	1+2ic_n o(1).
	\end{align}
	
	By combining \eqref{main-relation} with \eqref{stat-2}, the proof follows by taking $\widetilde{C}\coloneqq1+2ic_n o(1)$.
\end{proof}

The next result is an extension of Proposition \ref{Itoi-bound} when asymptotically vanishing perturbations are added to the random field Ginzburg-Landau model.
\begin{lemma1}\label{Lemma2}
	The expectation of the functions of the spin variables $\phi_{ x}$ and $\phi_{\bf x}$ have the following upper bounds:
	\begin{itemize}
		\item[(i)]
		$\mathbb{E}\langle\phi_{ x}^k\rangle_{\boldsymbol{\alpha}}
		\leqslant
		\widetilde{C} C_{2k}^{1/2}$, 
		\ \text{for all} \ ${x}\in V_n$;
		\item[(ii)]
		$\mathbb{E}\langle\phi_{\bf x}^k\rangle_{\boldsymbol{\alpha}}
		\leqslant
		\widetilde{C} C_{2kp}^{1/2}$ 
		and 
		
		$\mathbb{E}\langle\phi_{\bf x}^{2k}\rangle_{\boldsymbol{0}}
		\leqslant
		C_{2kp},
		$
		\ \text{for all} \ ${\bf x}\in \otimes_{l=2}^p V_n$, $p\geqslant 2$;
	\end{itemize}
	where $k$ is an arbitrary even integer.
\end{lemma1}
\begin{proof}
	Taking $F(\phi)=\phi_{ x}^k$ in Lemma \ref{prop-fund},	we have
	$
	\mathbb{E}\langle \phi_{ x}^k\rangle_{\boldsymbol{\alpha}}
	\leqslant
	\widetilde{C}\, 	
	\mathbb{E}^{1/2} 
	\langle \phi_{ x}^{2k} \,
	\rangle_{\boldsymbol{0}}.
	$
	By combining this inequality with Proposition \ref{Itoi-bound}, the proof of Item (i) follows.
	
	On the other hand, taking $F(\phi)=\phi_{\bf x}^k$ in Lemma \ref{prop-fund}, by Hölder's inequality, we get
	\begin{align*}
	\mathbb{E}\langle \phi_{\bf x}^k\rangle_{\boldsymbol{\alpha}}
	\leqslant
	\widetilde{C}\, 	
	\mathbb{E}^{1/2} 
	\langle \phi_{\bf x}^{2k} \,
	\rangle_{\boldsymbol{0}}
	\leqslant
	\widetilde{C}\, 
	\prod_{j=1}^{p}
	\mathbb{E}^{1/2p}
	\langle \phi_{x_j}^{2kp} \,
	\rangle_{\boldsymbol{0}}.
	\end{align*}
	
	By combining the above inequalities with Proposition \ref{Itoi-bound}, the proof of Item (ii) follows.
	%
	%
\end{proof}

The next result guarantees that the Hamiltonian 
$H_{n;\boldsymbol{\alpha}}(\phi)$ in \eqref{ham-1} is well-defined almost surely.
\begin{proposition}\label{prop-def-ham}
	The perturbing Hamiltonian $H^{\rm per}_{n;\boldsymbol{\alpha}}(\phi;\xi)$ in \eqref{Function-H}
	converges almost surely.
\end{proposition}
\begin{proof}
	By Lemma 3.6$''$ of \cite{Loeve-1951} and by definition of $H^{\rm per}_{n;\boldsymbol{\alpha}}(\phi;\xi)$,  it is enough to verify that
	\begin{align}\label{object}
	\sum_{p=2}^\infty
	\alpha_p 2^{-p}
	\mathbb{E}
	\langle
	\vert H_{n;p}(\phi;\xi)\vert
	\rangle_{\boldsymbol{\alpha}}
	<\infty.
	\end{align}
	
	Indeed,  taking $F(\phi)=\vert H_{n;p}(\phi;\xi)\vert$ in Lemma \ref{prop-fund}, by Proposition \ref{prop-bound-H},	we have
	\begin{align*}
	\mathbb{E}\langle \vert H_{n;p}(\phi;\xi)\vert\rangle_{\boldsymbol{\alpha}}
	\leqslant
	\widetilde{C} \, 
	\mathbb{E}^{1/2}
	\langle  [H_{n;p}(\phi;\xi)]^2 \,
	\rangle_{\boldsymbol{0}}
	\leqslant
	\widetilde{C}\sqrt{\vert V_n\vert} C_{2p}^{1/2}.
	\end{align*}
	Multiplying by $\alpha_p 2^{-p}$ and then summing  on $p\geqslant 2$ in the above inequalities, by regularity condition \eqref{regularity condition}, the statement in \eqref{object} follows.
	
	We thus complete the proof.
\end{proof}

\subsection{Properties of the Free energy}
For each $(\beta,h,u,r)\in (0,\infty)^3\times \mathbb{R}$, let 
\begin{align}\label{log}
\begin{array}{lll}
\psi_{n;\boldsymbol{\alpha}}= \psi_{n;\boldsymbol{\alpha}}(\beta, h,u,r) \coloneqq \dfrac{\log Z_{n;\boldsymbol{\alpha}}}{|V_n|};
&
& p_{n;\boldsymbol{\alpha}}=p_{n;\boldsymbol{\alpha}}(\beta, h,u,r)\coloneqq\mathbb{E}\psi_{n;\boldsymbol{\alpha}};
\end{array}
\end{align}
be the free energy and the average free energy, respectively.

The following result states that 
the perturbation term has a small influence in the computation of the expectation of the free energy $\psi_{n;\boldsymbol{\alpha}}$.
\begin{lemma1}\label{for-to-prove-lemma}
	We have
	\begin{align}\label{for-to-prove}
	p_{n;\boldsymbol{0}}
	\leqslant
	p_{n;\boldsymbol{\alpha}}
	\leqslant
	p_{n;\boldsymbol{0}}
	+o(1),
	\end{align}
	where $p_{n;\boldsymbol{\alpha}}$ is the average free energy given in \eqref{log}.
\end{lemma1}
\begin{proof}
	A simple computation shows that
	\begin{align} \label{key-1}
	\mathbb{E}\log Z_{n;\boldsymbol{\alpha}}
	&
	\stackrel{\eqref{partition-function}}{=}
	\mathbb{E}\log \int_{\mathbb{R}^{|V_n|}}  
	\exp\big(-	
	H_n(\phi;g)+H^{\rm per}_{n;\boldsymbol{\alpha}}(\phi;\xi)
	\big) \, {\rm D}\phi \nonumber
	\\[0,15cm]
	&=
	\mathbb{E}\log 
	Z_{n;\boldsymbol{0}}
	\big\langle 
	{\exp}\big({H^{\rm per}_{n;\boldsymbol{\alpha}}(\phi;\xi)}\big)
	\big\rangle_{\boldsymbol{0}}	
	\nonumber
	\\[0,15cm]
	&\leqslant
	\mathbb{E}\log 
	Z_{n;\boldsymbol{0}}
	+
	\mathbb{E}
	\big\langle 
	{\exp}\big({H^{\rm per}_{n;\boldsymbol{\alpha}}(\phi;\xi)}\big)
	\big\rangle_{\boldsymbol{0}}
	-1,
	\end{align}
	where in the last line we used the inequality $\log(x)\leqslant x-1$ for $x>0$.  
	
	Let $X_n$ be as in \eqref{def-X} and let $\widehat{\mathbb{E}}$ be the expectation over $\xi_{\bf x}$,  $p\geqslant 2$. Note that, ${H^{\rm per}_{n;\boldsymbol{\alpha}}(\phi;\xi)}=i c_n X_n$, and by \eqref{finite-exp}, that
	$
	\mathbb{E}
	\langle	
	\vert X_n\vert
	\rangle_{\boldsymbol{0}}
	<\infty.
	$
	Consequently, for $n\to\infty$, it follows that
	\begin{align}\label{ineq-Mink}
	\mathbb{E}
	\big\langle 
	{\exp}\big({H^{\rm per}_{n;\boldsymbol{\alpha}}(\phi;\xi)}\big)
	\big\rangle_{\boldsymbol{0}} 
	=
	\mathbb{E}
	\big\langle 
	{\exp}(ic_n X_n)
	\big\rangle_{\boldsymbol{0}} 
	&=
	1+ic_n \mathbb{E}
	\langle X_n\rangle_{\boldsymbol{0}} 
	+
	ic_n o(1)
	\nonumber
	\\[0,2cm]
	&=
	1+\mathbb{E}\langle {H^{\rm per}_{n;\boldsymbol{\alpha}}(\phi;\xi)}\rangle_{\boldsymbol{0}} 
	+
	ic_n o(1)
	\nonumber
	\\[0,2cm]
	&=
	1+\mathbb{E}
	\big\langle { \widehat{\mathbb{E}} H^{\rm per}_{n;\boldsymbol{\alpha}}(\phi;\xi)} \big\rangle_{\boldsymbol{0}} 
	+
	ic_n o(1)
	\nonumber
	\\[0,2cm]
	&=
	1+ic_n o(1),
	\end{align}
	where in the third equality we use Fubini's theorem to justify the change in order of expectations and in the last one that $\widehat{\mathbb{E}}H^{\rm per}_{n;\boldsymbol{\alpha}}(\phi;\xi) =0$.
	
	By \eqref{key-1} and \eqref{ineq-Mink}, we get
	\begin{align*}
	\mathbb{E}\log Z_{n;\boldsymbol{\alpha}}
	\leqslant
	\mathbb{E}\log 
	Z_{n;\boldsymbol{0}}
	+
	ic_n o(1).
	\end{align*}
	
	Dividing the above inequality by $\vert V_n\vert$ and using the definition in \eqref{log} of $p_{n;\boldsymbol{\alpha}}$, the proof of the upper bound \eqref{for-to-prove} of the average free energy $p_{n;\boldsymbol{\alpha}}$  follows.

	\bigskip 
	Now we get the lower bound of $p_{n;\boldsymbol{\alpha}}$ in \eqref{for-to-prove}. Indeed, a simple computation shows that
	\begin{align}
	{1\over \vert V_n\vert}\, \log Z_{n;\boldsymbol{\alpha}}
	-
	{1\over \vert V_n\vert}\, \log Z_{n;\boldsymbol{0}}
	&\stackrel{\eqref{partition-function}}{=}
	{1\over \vert V_n\vert}\, \log
	{\int_{\mathbb{R}^{|V_n|}}  
		\exp\big(-	
		H_n(\phi;g)+H^{\rm per}_{n;\boldsymbol{\alpha}}(\phi;\xi)
		\big) \, {\rm D}\phi	
		\over 
		\int_{\mathbb{R}^{|V_n|}}  
		\exp\big(-	H_n(\phi;g)\big) \, {\rm D}\phi
	}
	\nonumber
	\\[0,15cm]
	&=
	{1\over \vert V_n\vert}\, \log
	\big\langle 
	\exp\big(H^{\rm per}_{n;\boldsymbol{\alpha}}(\phi;\xi)
	\big) 
	\big\rangle_{\boldsymbol{0}} 	\nonumber
	\\[0,15cm]
	&\geqslant 
	{1\over \vert V_n\vert}\, 
	\big\langle { H^{\rm per}_{n;\boldsymbol{\alpha}}(\phi;\xi)} \big\rangle_{\boldsymbol{0}},
	\nonumber
	\end{align}
	where in the above inequality we used Jensen's inequality.
	Taking the expectation from both sides of the above inequality,
	we obtain 
	\begin{align*}
	p_{n;\boldsymbol{\alpha}}-p_{n;\boldsymbol{0}}
	\geqslant
	{1\over \vert V_n\vert}\, 
	\mathbb{E}
	\big\langle { H^{\rm per}_{n;\boldsymbol{\alpha}}(\phi;\xi)} \big\rangle_{\boldsymbol{0}}
	=
	{1\over \vert V_n\vert}\, 
	\mathbb{E}
	\big\langle { \widehat{\mathbb{E}} H^{\rm per}_{n;\boldsymbol{\alpha}}(\phi;\xi)} \big\rangle_{\boldsymbol{0}}
	= 0,
	\end{align*}	
	because $\widehat{\mathbb{E}}H^{\rm per}_{n;\boldsymbol{\alpha}}(\phi;\xi) =0$. Thus, we have complete the proof of lemma.
\end{proof}

\begin{lemma1}\label{lemma-exist-p}
	For every $(\beta, h, u, r)\in(0,\infty)^3\times \mathbb{R}$, the limit, defined in \eqref{log},
	\begin{align*}
	p_{\boldsymbol{\alpha}}
	=
	\lim_{n\to\infty}
	p_{n;\boldsymbol{\alpha}}(\beta, h,u,r)
	\end{align*}
	exists and is finite.
\end{lemma1}
\begin{proof}
	From Lemma \ref{for-to-prove-lemma}, $p_{\boldsymbol{\alpha}}
	=\lim_{n\to\infty}
	p_{n;\boldsymbol{0}}(\beta, h,u,r)$. But, it is known that $\lim_{n\to\infty}p_{n;\boldsymbol{0}}$ exists and is finite for every $(\beta, h, u, r)\in(0,\infty)^3\times \mathbb{R}$; see, e.g., Lemma II.2 of \cite{Itoi-2019}. This completes the proof.
\end{proof}

\begin{lemma1}\label{convexity}
	The limit $p_{\boldsymbol{\alpha}}$ is a convex function
	of $h$ for every fixed $\beta$, $u$ and $r$. The same statement is true for $\psi_{n;\boldsymbol{\alpha}}$ and $p_{n;\boldsymbol{\alpha}}$.
\end{lemma1}
\begin{proof}
	The proof is trivial and omitted.
\end{proof}

\subsection{Upper bound of Variance of Free energy}
%

Let $g=(g_x)$ and $g'=(g_x')$ be two disorders  consisting of i.i.d. standard Gaussian random variables.
For each $s\in[0,1]$, we define a new random field $G=(G_{x})$ as follows:
\begin{align*}
G_x\coloneqq \sqrt{s}\,g_x+\sqrt{1-s}\,g_x', \quad x\in V_n.
\end{align*}
Additionally, similarly to \cite{Itoi-2019}, we consider the following generating function:
\begin{align*}
\gamma_{n;\boldsymbol{\alpha}}(s)
\coloneqq 
\mathbb{E}\big[\mathbb{E}' \widehat{\mathbb{E}} \psi_{n;\boldsymbol{\alpha}}(G)\big]^2
=
\mathbb{E}\left[\mathbb{E}' \widehat{\mathbb{E}}
{1\over \vert V_n\vert}\,
\log
Z_{n;\boldsymbol{\alpha}}(G)
\right]^2,
\end{align*}
with
\begin{align*}
Z_{n;\boldsymbol{\alpha}}(G)\coloneqq \int_{\mathbb{R}^{|V_n|}}  
\exp\big(-	
H_n(\phi;G)
+
H^{\rm per}_{n;\boldsymbol{\alpha}}(\phi;\xi)
\big) \, {\rm D}\phi,
\end{align*}
where
$H_n(\phi;G)$ and $H^{\rm per}_{n;\boldsymbol{\alpha}}(\phi;\xi)$ are as in \eqref{Hamiltonian-g} and \eqref{Function-H}, respectively. Here,
$\mathbb{E}$ and $\mathbb{E}'$ denote expectation over $g$ and $g'$, respectively, and $\widehat{\mathbb{E}}$ is the expectation over $\xi_{\bf x}$, $p\geqslant 2$.
\begin{proposition}\label{bound-var-1}
	For any $(\beta, h, u, r)\in(0,\infty)^3\times \mathbb{R}$, any $k\geqslant 1$ integer and any $s\in[0,1]$, the $k$-th order derivative of $\gamma_{n;\boldsymbol{\alpha}}$, denoted by $\gamma^{(k)}_{n;\boldsymbol{\alpha}}$, is written as follows
	\begin{align*}
	\gamma^{(k)}_{n;\boldsymbol{\alpha}}(s)
	=
	\sum_{x_1,\ldots,x_k \in V_n}
	\mathbb{E}
	\left(
	\mathbb{E}' \widehat{\mathbb{E}} \,
	{\partial^k \psi_{n;\boldsymbol{\alpha}} \over \partial G_{x_k}\cdots  G_{x_1} } (G(s))
	\right)^2.
	\end{align*}
	Furthermore, for any $s_0\in[0,1)$,
	\begin{align*}
	&\gamma^{(k)}_{n;\boldsymbol{\alpha}}(s_0)
	\leqslant
	{(k-1)!\over (1-s_0)^{k-1}}\, {h^2
		\widetilde{C} C_{4}^{1/2}\over\vert V_n\vert};
	\\[0,1cm]
	&
	\gamma'_{n;\boldsymbol{\alpha}}(1)
	\leqslant
	{h^2 \widetilde{C} C_{4}^{1/2}\over\vert V_n\vert}.
	\end{align*}
\end{proposition}
\begin{proof}
	The proof follows by combining the same steps of the proof of Lemma II.5 in \cite{Itoi-2019} with Lemma \ref{Lemma2}-Item (i). Then it is omitted.
\end{proof}

\begin{lemma1}\label{upper-bound-variance}
	For each $(\beta, h, u, r)\in(0,\infty)^3\times \mathbb{R}$,
	the variance of the free-energy density $\psi_{n;\boldsymbol{\alpha}}$ in \eqref{log} is bounded from the above as follows:
	\begin{align*}
	{\rm Var}(\psi_{n;\boldsymbol{\alpha}})
	\leqslant
	{h^2 \widetilde{C} C_{4}^{1/2}\over\vert V_n\vert}.
	\end{align*}
\end{lemma1}
\begin{proof}
	The proof of this lemma follows by combining the inequality
	\begin{align*}
	{\rm Var}(\psi_{n;\boldsymbol{\alpha}})
	=
	\int_{0}^{1}\gamma'_{n;\boldsymbol{\alpha}}(s)\, {\rm d} s
	\leqslant 
	\gamma'_{n;\boldsymbol{\alpha}}(1)
	\end{align*}
	with Proposition \ref{bound-var-1}.	
\end{proof}

\subsection{Properties of the Energy function due to disorder} 
For any $n\geqslant 1$, let us define
\begin{align}\label{delta-n}
\Delta_{n}
=
{1\over\vert V_n\vert}\, 
\sum_{x\in V_n} g_{x} \phi_{x}
\end{align}
the part of the energy function \eqref{ham-1} due to the first disorder $g=(g_x)$.

Let ${\mathcal A}$ be the countable set of all $(\beta, h, u, r)\in (0,\infty)^3\times\mathbb{R}$
such that 
$\textstyle{\partial p_{\boldsymbol{\alpha}}\over \partial h^-}(\beta, h, u, r)\neq {\partial p_{\boldsymbol{\alpha}}\over \partial h^+}(\beta, h, u, r)$ and 
$\textstyle{\partial p_{\boldsymbol{\alpha}}\over \partial r^-}(\beta, h, u, r)\neq {\partial p_{\boldsymbol{\alpha}}\over \partial r^+}(\beta, h, u, r)$.

\begin{proposition}\label{cond-main} 
	For any $(\beta, h, u, r)\in {\mathcal A}^c$,
	\begin{itemize}
		\item[1)] 
		$\displaystyle
		\mathbb{E}\langle \Delta_{n}\rangle_{\boldsymbol{\alpha}} 
		\underset{ n\rightarrow\infty}{\xrightarrow{\hspace*{0,9cm}}}
		{\partial p_{\boldsymbol{\alpha}}\over\partial h}(\beta, h, u, r)
		$ and
		
		$
		\mathbb{E}\big|\smallavg{\Delta_{n}}_{\boldsymbol{\alpha}} -\mathbb{E}\langle\Delta_{n}\rangle_{\boldsymbol{\alpha}}\big| \underset{ n\rightarrow\infty}{\xrightarrow{\hspace*{0,9cm}}} 0;
		$
		\item[2)]
		$\displaystyle
		\mathbb{E}\langle R_{1,1}\rangle_{\boldsymbol{\alpha}} 
		\underset{ n\rightarrow\infty}{\xrightarrow{\hspace*{0,9cm}}}
		{\partial p_{\boldsymbol{\alpha}}\over\partial r}(\beta, h, u, r)
		$ and
		
		$
		\mathbb{E}\big|\smallavg{R_{1,1}}_{\boldsymbol{\alpha}} -\mathbb{E}\langle R_{1,1}\rangle_{\boldsymbol{\alpha}}\big| \underset{ n\rightarrow\infty}{\xrightarrow{\hspace*{0,9cm}}} 0.
		$
	\end{itemize}  
\end{proposition}
\begin{proof}
	The proof of this proposition is quite standard when we have the validity of the following three fundamental ingredients; see, e.g., Lemma 2.7 of \cite{chatterjee2015disorder}:
	\begin{itemize}
		\item the convexity of $\psi_{n;\boldsymbol{\alpha}}$, which is guaranteed by Lemma \ref{convexity};
		\item the variance of $\log Z_{n;\boldsymbol{\alpha}}$ doesn't grow faster than $\vert V_n\vert$, which is guaranteed by Lemma \ref{upper-bound-variance}; and that
		\item the limit $p_{\boldsymbol{\alpha}}= \lim_{n\ra\infty} p_{n;\boldsymbol{\alpha}}$ exists and is finite, which is guaranteed by Lemma \ref{lemma-exist-p}.	
	\end{itemize}
	Therefore, the proof follows.
\end{proof}

\begin{proposition}\label{prop-main-1}
	Let $F(\phi):\mathbb{R}^{V_n} \to \mathbb{R}$ be a Borel function. If the following hold: 
	\begin{itemize}
		\item[(a)] $\mathbb{E}(\langle F^2(\phi)\rangle_{\boldsymbol{\alpha}}-\langle F(\phi)\rangle^2_{\boldsymbol{\alpha}}) 
		=
		\mathbb{E}(\langle F^2(\phi)\rangle_{\boldsymbol{\alpha}}-\mathbb{E}\langle F(\phi)\rangle^2_{\boldsymbol{\alpha}})
		\underset{ n\rightarrow\infty}{\xrightarrow{\hspace*{0,9cm}}} 0$;
		\item[(b)] $\mathbb{E}|\smallavg{F(\phi)}_{\boldsymbol{\alpha}} -\mathbb{E}\langle F(\phi)\rangle_{\boldsymbol{\alpha}}| \underset{ n\rightarrow\infty}{\xrightarrow{\hspace*{0,9cm}}} 0$.
	\end{itemize}
	Then
	\[
	\mathbb{E}\big\langle|F(\phi) - \mathbb{E}\langle F(\phi)\rangle_{\boldsymbol{\alpha}}  | \big\rangle_{\boldsymbol{\alpha}}  
	\underset{ n\rightarrow\infty}{\xrightarrow{\hspace*{0,9cm}}} 0.
	\]
\end{proposition}
\begin{proof} The proof is immediate since, by triangle and Lyapunov inequalities, 
	\begin{align*}
	\mathbb{E}\big\langle|F(\phi) - \mathbb{E}\langle F(\phi)\rangle_{\boldsymbol{\alpha}}  | \big\rangle_{\boldsymbol{\alpha}}
	&\leqslant
	\mathbb{E}\langle|F(\phi)-\langle F(\phi)\rangle_{\boldsymbol{\alpha}}|\rangle_{\boldsymbol{\alpha}}
	+
	\mathbb{E}\big|\smallavg{F(\phi)}_{\boldsymbol{\alpha}} -\mathbb{E}\langle F(\phi)\rangle_{\boldsymbol{\alpha}}\big|
	\\[0,2cm]
	&\leqslant
	\sqrt{\mathbb{E}(\langle F^2(\phi)\rangle_{\boldsymbol{\alpha}}-\langle F(\phi)\rangle^2_{\boldsymbol{\alpha}})}
	+
	\mathbb{E}\big|\smallavg{F(\phi)}_{\boldsymbol{\alpha}} -\mathbb{E}\langle F(\phi)\rangle_{\boldsymbol{\alpha}}\big|.	
	\end{align*}	
\end{proof}

\begin{lemma1}\label{lemma-main}
	For any $(\beta, h, u, r)\in {\mathcal A}^c$,
	\[
	\mathbb{E}\big\langle|\Delta_n - \mathbb{E}\langle\Delta_n\rangle_{\boldsymbol{\alpha}} | \big\rangle_{\boldsymbol{\alpha}}  
	\underset{ n\rightarrow\infty}{\xrightarrow{\hspace*{0,9cm}}} 0.
	\]
\end{lemma1}
\begin{proof}
	Note that, by combining Proposition \ref{cond-main}-Item 1 and Proposition \ref{prop-main-1}, with $F(\phi)=\Delta_n$, it is sufficient to verify that
	\begin{align}\label{obj-conv}
	\mathbb{E}\big(\smallavg{\Delta_n^2}_{\boldsymbol{\alpha}}-\smallavg{\Delta_n}^2_{\boldsymbol{\alpha}}\big)
	\underset{ n\rightarrow\infty}{\xrightarrow{\hspace*{0,9cm}}} 0.
	\end{align}
	
	Indeed,	
	as sub-product of proof of Lemma II.7 in \cite{Itoi-2019} we have the following inequality:
	\begin{align*}
	\mathbb{E}\big(\smallavg{\Delta_n^2}_{\boldsymbol{\alpha}}-\smallavg{\Delta_n}^2_{\boldsymbol{\alpha}}\big)
	&\leqslant
	{1\over h^2}\,
	\Biggl(\,
	\sum_{x,y\in V_n}\biggl(\mathbb{E} {\partial^4 \psi_{n;\boldsymbol{\alpha}}\over \partial g_x^2 \partial g_y^2}\biggr)^2 
	\Biggr)^{1/2}
	+
	{1\over \vert V_n\vert^2}\,
	\sum_{x\in V_n} 
	\mathbb{E}\big(\smallavg{\phi_x^2}_{\boldsymbol{\alpha}}- \smallavg{\phi_x}_{\boldsymbol{\alpha}}^2\big).
	\end{align*}
	Using Proposition \eqref{bound-var-1}, with $k=4$, $x_1=x_2=x$ and $x_3=x_4=y$:
	\begin{align*}
	\sum_{x,y\in V_n}\biggl(\mathbb{E} {\partial^4 \psi_{n;\boldsymbol{\alpha}}\over \partial g_x^2 \partial g_y^2}\biggr)^2
	=
	\gamma''''_{n;\boldsymbol{\alpha}}(0),
	\end{align*}
	we written the last inequality as follows
	\begin{align*}
	\mathbb{E}\big(\smallavg{\Delta_n^2}_{\boldsymbol{\alpha}}-\smallavg{\Delta_n}^2_{\boldsymbol{\alpha}}\big)
	&\leqslant
	{1\over h^2}\,
	\big[\gamma''''_{n;\boldsymbol{\alpha}}(0) \big]^{1/2}
	+
	{1\over \vert V_n\vert^2}\,
	\sum_{x\in V_n} 
	\mathbb{E}\big(\smallavg{\phi_x^2}_{\boldsymbol{\alpha}}- \smallavg{\phi_x}_{\boldsymbol{\alpha}}^2\big).
	\end{align*}
	By Proposition \ref{bound-var-1} and Lemma \ref{Lemma2}-Item (i), the right-hand expression is at most
	\begin{align*}
	{1\over h}\,
	{\sqrt{6 \widetilde{C} C_{4}^{1/2} } \over \sqrt{\vert V_n\vert}} 
	+
	{\widetilde{C} C_{4}^{1/2} \over\vert V_n\vert}.
	\end{align*}
	Therefore,
	\begin{align*}
	\mathbb{E}\big(\smallavg{\Delta_n^2}_{\boldsymbol{\alpha}}-\smallavg{\Delta_n}^2_{\boldsymbol{\alpha}}\big)
	\leqslant
	{1\over h}\,
	{\sqrt{6 \widetilde{C} C_{4}^{1/2} } \over \sqrt{\vert V_n\vert}} 
	+
	{\widetilde{C} C_{4}^{1/2} \over\vert V_n\vert}.
	\end{align*}
	Letting $n\to\infty$ in the above inequality, \eqref{obj-conv} follows.
\end{proof}

\begin{lemma1}\label{lemma-main-1}
	For any $(\beta, h, u, r)\in {\mathcal A}^c$,
	\[
	\mathbb{E}\big\langle|R_{1,1} - \mathbb{E}\langle R_{1,1}\rangle_{\boldsymbol{\alpha}}  | \big\rangle_{\boldsymbol{\alpha}}  
	\underset{ n\rightarrow\infty}{\xrightarrow{\hspace*{0,9cm}}} 0.
	\]
\end{lemma1}
\begin{proof} 
	Note that, by combining Proposition \ref{cond-main}-Item 2 and Proposition \ref{prop-main-1}, with $F(\phi)=R_{1,1}$, it is sufficient to verify that
	\begin{align}\label{obj-conv-1}
	\mathbb{E}\big(\smallavg{R_{1,1}^2}_{\boldsymbol{\alpha}}-\smallavg{R_{1,1}}^2_{\boldsymbol{\alpha}}\big)
	=
	{1\over\vert V_n\vert^2}\, \sum_{x,y\in V_n} \mathbb{E}\langle\phi_x^2;\phi_y^2\rangle_{\boldsymbol{\alpha}}
	\underset{ n\rightarrow\infty}{\xrightarrow{\hspace*{0,9cm}}} 0,
	\end{align}	
	where $\smallavg{\phi_x;\phi_y}_{\boldsymbol{\alpha}}
	\coloneqq
	\smallavg{\phi_x \phi_y}_{\boldsymbol{\alpha}}-
	\smallavg{\phi_x}_{\boldsymbol{\alpha}}\smallavg{\phi_y}_{\boldsymbol{\alpha}}$ is the truncated two-point correlation.
	But \eqref{obj-conv-1} results from a straightforward adaptation of the arguments of the proof of Lemma II.13 in \cite{Itoi-2019}, so this one is omitted.
\end{proof}
%

%
%

%
%
%
\subsection{The Ghirlanda-Guerra identities}
Lemmas \ref{lemma-main} and \ref{lemma-main-1} enable us to derive the Ghirlanda-Guerra identities \cite{YTchen2019,ghirlanda1998general,panchenko2010,panchenko2011ghirlanda,PANCHENKO2010189,book-talagrand,talagrand2003spin} (also sometimes called the Aizenman-Contucci identities \cite{aizenman1998stability}) at almost all $(\beta, h, u, r)$ for the random mixed-spin Ginzburg-Landau model.
As usual, details are presented below for
the sake of completeness.

Take any integer $m\geqslant 2$ and let $\phi^1,\ldots, \phi^m, \phi^{m+1}$ 
denote $m+1$ spin configurations drawn independently from the Gibbs measure.  
Let $R_{\ell,\ell'}$ be the overlap between $\phi^\ell$ and $\phi^{\ell'}$ defined in \eqref{overlap}, 
with $\ell,\ell'=1,\ldots,m+1$. 
Let $f:\mathbb{R}^{m(m-1)/2}\to[-1,1]$ be a bounded measurable function
of these overlaps that not change with $n$.
%
\begin{lemma1}\label{ggid}
	Consider the random field mixed-spin Ginzburg-Landau model  defined by the Gibbs measure in \eqref{gibbs-measure2}. 
	Then, if $f$ is as above, the following identity   
	\begin{align*} 
	\mathbb{E}\smallavg{ f R_{1,m+1}}_{\boldsymbol{\alpha}}
	- 
	\frac{1}{m}  \,
	\mathbb{E}\smallavg{ f}_{\boldsymbol{\alpha}}
	\mathbb{E}\smallavg{R_{1,2}}_{\boldsymbol{\alpha}}
	- 
	\frac{1}{m}\,
	\sum_{s=2}^m \mathbb{E}\smallavg{ f R_{1,s}}_{\boldsymbol{\alpha}}
	\underset{ n\rightarrow\infty}{\xrightarrow{\hspace*{0,9cm}}} 0
	\end{align*}
	is satisfied at almost all $(\beta, h, u, r)$.
\end{lemma1}
\begin{proof}
	Since $\|f\|_\infty\leqslant 1$, from Lemmas \ref{lemma-main} and \ref{lemma-main-1}, it follows that 
	\begin{align}\label{first-part}
	&\left|
	\mathbb{E}\big\langle\Delta_{n}(\phi^1) f\rangle_{\boldsymbol{\alpha}}
	-
	\mathbb{E}\big\langle\Delta_{n}(\phi^1)\big\rangle_{\boldsymbol{\alpha}}
	\mathbb{E}\big\langle f\big\rangle_{\boldsymbol{\alpha}}
	\right|
	\leqslant 
	\mathbb{E}\big\langle
	|\Delta_n - \mathbb{E}\langle\Delta_n\rangle_{\boldsymbol{\alpha}} | \big\rangle_{\boldsymbol{\alpha}}
	\underset{ n\rightarrow\infty}{\xrightarrow{\hspace*{0,9cm}}} 0;
	\\[0,2cm]
	&
	\left|
	\mathbb{E}\big\langle R_{1,1} f\rangle_{\boldsymbol{\alpha}}
	-
	\mathbb{E}\big\langle R_{1,1}\big\rangle_{\boldsymbol{\alpha}}
	\mathbb{E}\big\langle f\big\rangle_{\boldsymbol{\alpha}}
	\right|
	\leqslant 
	\mathbb{E}\big\langle
	|R_{1,1} - \mathbb{E}\langle R_{1,1}\rangle_{\boldsymbol{\alpha}} | \big\rangle_{\boldsymbol{\alpha}}
	\underset{ n\rightarrow\infty}{\xrightarrow{\hspace*{0,9cm}}} 0; \label{first-part-1}
	\end{align}
	where $\Delta_{n}$ is as in \eqref{delta-n}.
	
	Let $\widehat{\mathbb{E}}$ be the expectation over $\xi_{\bf x}$, $p\geqslant 2$. Integration by parts gives
	\begin{align}\label{int-parts}
	\mathbb{E}\big\langle\Delta_n(\phi^1) f\big\rangle_{\boldsymbol{\alpha}}
	&= 
	{1\over \vert V_n\vert}\,
	\sum_{x\in V_n} 
	\mathbb{E}(g_x \widehat{\mathbb{E}}\langle\phi_x^{1} f\rangle_{\boldsymbol{\alpha}})
	\nonumber
	\\[0,15cm]
	&
	=
	{1\over \vert V_n\vert}\,
	\sum_{x\in V_n} 
	\mathbb{E}\biggl( {\partial \widehat{\mathbb{E}}\langle\phi_x^{1} f \rangle_{\boldsymbol{\alpha}} \over \partial g_x}\biggr)
	\nonumber
	\\[0,15cm]
	&
	=
	h
	\mathbb{E}\biggl\langle\Big(\sum_{\ell=1}^{m}R_{1,\ell}-mR_{1,m+1}\Big)f \biggr\rangle_{\boldsymbol{\alpha}}.
	\end{align}
	In the particular case where $f=1$ and $m=1$, 
	\begin{align*}
	\mathbb{E}\big\langle\Delta_n(\phi^1)\big\rangle_{\boldsymbol{\alpha}}
	=
	h
	\mathbb{E}\langle R_{1,1}-R_{1,2}\rangle_{\boldsymbol{\alpha}}.
	\end{align*}
	By combining the above identity with \eqref{first-part} and \eqref{int-parts}, we have
	\begin{align*}
	\sup_{\|f\|_\infty\leqslant 1}
	\biggr|
	\mathbb{E}\langle R_{1,1}-R_{1,2}\rangle_{\boldsymbol{\alpha}} \mathbb{E}\langle f\rangle_{\boldsymbol{\alpha}}
	-
	\mathbb{E}\biggl\langle\Big(\sum_{\ell=1}^{m}R_{1,\ell}-mR_{1,m+1}\Big)f \biggr\rangle_{\boldsymbol{\alpha}}
	\biggr|
	\underset{ n\rightarrow\infty}{\xrightarrow{\hspace*{0,9cm}}} 0.
	\end{align*}
	Finally, by applying \eqref{first-part-1} the proof of lemma follows.
\end{proof}

\subsection{The self-averaging of the spin overlap}
\begin{proposition}\label{bound-truncated-quadratic}
	For any $(\beta, h, u, r)\in(0,\infty)^3\times \mathbb{R}$,
	\begin{align*}
	\sum_{x,y\in V_n} 
	(\mathbb{E}\smallavg{\phi_x;\phi_y}_{\boldsymbol{\alpha}})^2
	\leqslant 
	{\widetilde{C} C_{4}^{1/2}\over h^2}\,
	\vert V_n\vert ,
	\end{align*}
	where $\smallavg{\phi_x;\phi_y}_{\boldsymbol{\alpha}}$ is the truncated two-point correlation.
\end{proposition}
\begin{proof}
	It is well-known that
	\begin{align*}
	{\partial^2 \psi_{n;\boldsymbol{\alpha}}\over \partial g_x \partial g_y}
	=
	{h^2\over \vert V_n\vert}\,
	\smallavg{\phi_x;\phi_y}_{\boldsymbol{\alpha}}.
	\end{align*}
	Hence,
	\begin{align*}
	\sum_{x,y\in V_n} 
	(\mathbb{E}\smallavg{\phi_x;\phi_y}_{\boldsymbol{\alpha}})^2
	=
	{\vert V_n\vert^2\over h^4}\,
	\sum_{x,y\in V_n} 
	\biggl(
	\mathbb{E}{\partial^2 \psi_{n;\boldsymbol{\alpha}}\over \partial g_x \partial g_y}
	\biggr)^2
	=
	{\vert V_n\vert^2\over h^4}\, \gamma''_{n;\boldsymbol{\alpha}}(0).
	\end{align*}
	Then the proof follows by applying Proposition \ref{bound-var-1}.
\end{proof}

\begin{lemma1}[Self-averaging of the overlap]
	\label{mainlmm-il}
	For any $(\beta, h, u, r)\in (0,\infty)^3\times \mathbb{R}$,
	\begin{align*}
	&\mathbb{E}\big(\smallavg{R_{1,2}^2}_{\boldsymbol{\alpha}} 
	- 
	\smallavg{R_{1,2}}^2_{\boldsymbol{\alpha}} 
	\big)
	=
	\mathbb{E}\big\langle R_{1,2}-\smallavg{R_{1,2}}_{\boldsymbol{\alpha}} \big\rangle^2_{\boldsymbol{\alpha}} 
	\underset{ n\rightarrow\infty}{\xrightarrow{\hspace*{0,9cm}}} 0.
	\end{align*}
\end{lemma1}
\begin{proof}
	As sub-product of proof of Lemma II.12  in \cite{Itoi-2019}, the following inequality remain valid:
	\begin{align*}
	\mathbb{E}\big(\smallavg{R_{1,2}^2}_{\boldsymbol{\alpha}} 
	- 
	\smallavg{R_{1,2}}^2_{\boldsymbol{\alpha}} 
	\big)
	\leqslant
	{2\over\vert V_n\vert^2}\,
	\Biggl(\,
	\sum_{x,y\in V_n} 
	\mathbb{E} \smallavg{\phi_x;\phi_y}_{\boldsymbol{\alpha}}^2
	\sum_{x,y\in V_n} 
	\mathbb{E}^{1/2}
	\smallavg{\phi_x^4}_{\boldsymbol{\alpha}}\,
	\mathbb{E}^{1/2}
	\smallavg{\phi_y^4}_{\boldsymbol{\alpha}}
	\Biggr)^{1/2},
	\end{align*}
	where the FKG inequality plays an important role in achieving this one.
	By Lemma \ref{Lemma2}-Item (i), the right-hand expression is at most
	\begin{align*}
	2
	\sqrt{
		\widetilde{C} C_{8}^{1/2} 
	}
	\Biggl(
	{1\over\vert V_n\vert^2}
	\sum_{x,y\in V_n} 
	\mathbb{E} \smallavg{\phi_x;\phi_y}_{\boldsymbol{\alpha}}^2
	\Biggr)^{1/2}.
	\end{align*}
	Therefore,
	\begin{align}\label{ineq-1}
	\mathbb{E}\big(\smallavg{R_{1,2}^2}_{\boldsymbol{\alpha}} 
	- 
	\smallavg{R_{1,2}}^2_{\boldsymbol{\alpha}} 
	\big)
	\leqslant
	2
	\sqrt{
		\widetilde{C} C_{8}^{1/2} 
	}
	\Biggl(
	{1\over\vert V_n\vert^2}
	\sum_{x,y\in V_n} 
	\mathbb{E} \smallavg{\phi_x;\phi_y}_{\boldsymbol{\alpha}}^2
	\Biggr)^{1/2}.
	\end{align}
	
	On the other hand, again, as sub-product of proof of Lemma II.11  in \cite{Itoi-2019}, we have
	\begin{align*}
	{1\over\vert V_n\vert^2}\,
	\sum_{x,y\in V_n} 
	\mathbb{E} \smallavg{\phi_x;\phi_y}_{\boldsymbol{\alpha}}^2
	&\leqslant
	{C'\over \vert V_n\vert^2}\,
	\Biggl(\,
	\sum_{x,y\in V_n} 
	(\mathbb{E}\smallavg{\phi_x;\phi_y}_{\boldsymbol{\alpha}})^2  \vert V_n\vert^2
	\Biggr)^{1/2}
	\nonumber
	\\[0,2cm] 
	&+
	4
	\Biggl(
	{1\over C'}\,
	\sup_{x\in V_n}\mathbb{E}\smallavg{\phi_x^8}_{\boldsymbol{\alpha}} \
	{1\over \vert V_n\vert^2}\,
	\biggl( \,
	\sum_{x,y\in V_n} 
	(\mathbb{E}\smallavg{\phi_x;\phi_y}_{\boldsymbol{\alpha}})^2  \vert V_n\vert^2
	\biggr)^{1/2} \,
	\Biggr)^{1/2},
	\end{align*}
	where $C'$ is a positive constant. By using Proposition \ref{bound-truncated-quadratic} and Lemma \ref{Lemma2}-Item (i), the above inequality is 
	\begin{align*}
	\leqslant
	{C'\over h}\,
	{\sqrt{\widetilde{C} C_{4}^{1/2}}\over \sqrt{\vert V_n\vert} }
	+
	4\ 	
	\sqrt{\widetilde{C} C_{16}^{1/2} \over C' h }
	\,	
	{\sqrt[4]{\widetilde{C} C_{4}^{1/2}}\over \sqrt[4]{\vert V_n\vert} }. 
	\end{align*}
	Hence,
	\begin{align*}
	{1\over\vert V_n\vert^2}\,
	\sum_{x,y\in V_n} 
	\mathbb{E} \smallavg{\phi_x;\phi_y}_{\boldsymbol{\alpha}}^2
	\leqslant
	{C'\over h}\,
	{\sqrt{\widetilde{C} C_{4}^{1/2}}\over \sqrt{\vert V_n\vert} }
	+
	4\ 	
	\sqrt{\widetilde{C} C_{16}^{1/2} \over C' h }
	\,	
	{\sqrt[4]{\widetilde{C} C_{4}^{1/2}}\over \sqrt[4]{\vert V_n\vert} }.
	\end{align*}
	
	By combining the above inequality with \eqref{ineq-1} and then letting $n\to\infty$, the proof follows.
\end{proof}

%
%
%
\subsection{Concluding the Proof of Theorem \ref{rsbthm}}
Given results like the self-averaging of the overlap (Lemma \ref{mainlmm-il}) and the Ghirlanda-Guerra identities (Lemma \ref{ggid}), the derivation of the absence of replica symmetry breaking for the random field mixed-spin Ginzburg-Landau model is quite standard; see, e.g., references \cite{chatterjee2015disorder,Itoi2019,Itoi-2019}. \qed
\section*{Acknowledgements}
\noindent
We would like to thank to H. Saulo and J. Roldan for many valuable comments and careful reading of this manuscript.
This study was financed in part by the Coordena\c{c}\~{a}o de Aperfei\c{c}oamento de Pessoal de N\'{i}vel Superior - Brasil (CAPES) - Finance Code 001.




\end{document}